\documentclass[a4paper]{article}

\usepackage{amsmath}
\usepackage{amsthm}
\usepackage{amsfonts}
\usepackage{enumerate}
\usepackage{dsfont}
\usepackage{eepic}
\usepackage{graphicx}
\usepackage{url,eucal}
\usepackage{psfrag}

\newcommand{\reals}{\mathbb{R}}

\newcommand{\complex}{\mathbb{C}}

\newcommand{\integers}{\mathbb{Z}}

\newcommand{\bracket}[1]{\left[#1\right]}

\newcommand{\bracketb}[1]{\Big[#1\Big]}
\newcommand{\bracketc}[1]{\bigg[#1\bigg]}

\newcommand{\angles}[1]{\left\langle #1 \right\rangle}
\newcommand{\pb}[1]{\left\{#1\right\}}
\newcommand{\com}[1]{\left[#1\right]}

\newcommand{\para}[1]{\left(#1\right)}
\newcommand{\paraa}[1]{\big(#1\big)}

\newcommand{\parac}[1]{\bigg(#1\bigg)}

\newtheorem{theorem}{Theorem}[section]
\newtheorem{proposition}[theorem]{Proposition}
\newtheorem{lemma}[theorem]{Lemma}
\newtheorem{definition}[theorem]{Definition}
\newtheorem{corollary}[theorem]{Corollary}

\newcommand{\spacearound}[1]{\quad#1\quad}
\newcommand{\equivalent}{\spacearound{\Longleftrightarrow}}

\newcommand{\Com}[2]{\big[#1,#2\big]}
\newcommand{\st}{\,:\,}
\newcommand{\Rad}{\operatorname{Rad}}
\newcommand{\diag}{\operatorname{diag}}
\newcommand{\sgn}{\operatorname{sgn}}
\newcommand{\half}{\frac{1}{2}}
\newcommand{\phih}{\hat{\phi}}

\renewcommand{\tilde}[1]{\widetilde{#1}}
\newcommand{\su}{\mathfrak{su}(2)}
\newcommand{\glnc}{\mathfrak{gl}(N,\mathbb{C})}
\newcommand{\sph}{\mathbb{S}^2}
\newcommand{\cote}{\operatorname{Cote}_x}
\newcommand{\grad}{\operatorname{grad}}

\renewcommand{\mid}{\mathds{1}}
\newcommand{\h}{\hbar}
\newcommand{\C}{C}
\newcommand{\Cmuh}{\C(\mu,\h)}
\newcommand{\Wd}{W^\dagger}
\newcommand{\Wb}{\overline{W}}
\newcommand{\Wh}{\hat{W}}
\newcommand{\Dt}{\tilde{D}}
\newcommand{\dt}{\tilde{d}}
\newcommand{\et}{\tilde{e}}
\newcommand{\Ch}{\hat{C}}
\newcommand{\A}{\mathcal{A}}
\newcommand{\xv}{\vec{x}}
\newcommand{\xh}{\hat{x}}
\newcommand{\vphi}{\varphi}

\newcommand{\Rb}{\mathbb{R}}
\newcommand{\Sb}{\mathbb{S}}

\newcommand{\Zb}{\mathbb{Z}}
\newcommand{\fatone}{\mid}
\newcommand{\til}[1]{\widetilde{#1}}
\renewcommand{\and}{{\quad\text{and}\quad}}
\newcommand{\ou}{\quad\text{where}\quad}

\newcommand{\mat}{\mathcal{M}}
\newcommand{\smat}[4]{\begin{pmatrix}
			0 & #1 & 0 & \cdots & 0\\
			0 & 0 & #2 & \ddots & 0\\
		        \vdots & \vdots & \vdots & \ddots & \vdots\\
			0 & 0 & 0 & \cdots & #3\\
			#4 & 0 & 0 & \cdots & 0
                      \end{pmatrix}}
\newcommand{\spsmat}[4]{\begin{pmatrix}
    0      & #1    & 0      & \cdots & 0      & #4 \\
    #1    & 0      & #2    & \cdots & 0      & 0   \\
    0      & #2    & 0      & \ddots & 0      & 0   \\
    \vdots & \ddots & \ddots & \ddots & \ddots & \vdots\\
    0      & 0      & \cdots & \ddots & 0      & #3 \\
    #4    & 0      & \cdots & 0      & #3 & 0
  \end{pmatrix}}
\newcommand{\smsmat}[4]{\begin{pmatrix}
    0      & #1    & 0      & \cdots & 0      & -#4 \\
    -#1    & 0      & #2    & \cdots & 0      & 0   \\
    0      & -#2    & 0      & \ddots & 0      & 0   \\
    \vdots & \ddots & \ddots & \ddots & \ddots & \vdots\\
    0      & 0      & \cdots & \ddots & 0      & #3 \\
    #4    & 0      & \cdots & 0      & -#3 & 0
  \end{pmatrix}}

\newtheorem{thm}[theorem]{Theorem}
\newtheorem{prop}[theorem]{Proposition}
\newtheorem{lem}[theorem]{Lemma}

\theoremstyle{remark}
\newtheorem{rmk}[theorem]{Remark}

\numberwithin{equation}{section}

\makeatletter
\newcommand{\ps@draft}{%
\renewcommand{\@oddhead}{\hfill\textit{Preliminary version of \today}\hfill}}%
\makeatother

\title{\flushleft{Fuzzy Riemann Surfaces}}
\author{\flushleft{Joakim Arnlind, Martin Bordemann,
Laurent Hofer} \\ \flushleft{Jens Hoppe, Hidehiko Shimada}}

\begin{document}


\noindent \textbf{\huge Noncommutative Riemann Surfaces}
\vspace{0.3cm}\\
\noindent \textrm{ Joakim Arnlind$^1$, Martin Bordemann$^2$, Laurent Hofer$^2$,\\ 
  Jens Hoppe$^1$ and Hidehiko Shimada$^3$}
\vspace{0.5cm}\\
\noindent $^1$\,{\small Dept. of Mathematics, KTH, S-10044 Stockholm.\\
\hspace*{0.22cm}\url{jarnlind@math.kth.se}, \url{hoppe@math.kth.se}}\\
$^2$\,{\small Laboratoire de MIA, 4, rue des Fr\`eres Lumi\`ere, Universit\'e de Haute-Alsace, F-68093 Mulhouse.\\
\hspace*{0.22cm}\url{martin.bordemann@uha.fr}, \url{laurent.hofer@uha.fr}}\\
$^3$\,{\small Max Planck Institute for Gravitational Physics, Am M\"uhlenberg 1, D-14476 Golm.\\
\hspace*{0.22cm}\url{hidehiko.shimada@aei.mpg.de}}

\begin{abstract}
  \noindent We introduce \emph{C-Algebras} of compact Riemann surfaces
  $\Sigma$ as non-commutative analogues of the Poisson algebra of
  smooth functions on $\Sigma$.  Representations of these algebras
  give rise to sequences of matrix-algebras for which
  matrix-commutators converge to Poisson-brackets as $N\to\infty$.
  For a particular class of surfaces, nicely interpolating between
  spheres and tori, we completely characterize (even for the
  intermediate singular surface) all finite dimensional
  representations of the corresponding \emph{C-algebras}.
\end{abstract}


\section*{Introduction}

Attaching sequences of matrix algebras to a given manifold $M$ to
describe a non-commutative and approximate version of its ring of
smooth functions has now become a rather important tool in
non-commutative field theory: more precisely, for each positive integer
$N$ let $Q_N:\mathcal{C}^\infty (M,\mathbb{C})\rightarrow
M_{N,N}(\mathbb{C})$ be a complex linear surjective map of the ring of
smooth functions on $M$ into the space of all complex $N\times
N$-matrices such that products of functions are approximately mapped
to products of matrices in the limit $N\to \infty$. In almost all
cases, $\mathcal{C}^\infty (M,\mathbb{C})$ carries a Poisson bracket
$\{~,~\}$ (for instance if $M$ is symplectic such as every orientable
Riemann surface), and one further demands that Poisson brackets are
approximately mapped to matrix commutators in the limit $N\to \infty$.

For the $2$-sphere $\sph$ this was done in 1982 in \cite{Hop82}: using
the classical fact that the space of all spherical harmonics of fixed
$l$ is in bijection with the space of all harmonic polynomials in
$\mathbb{R}^3$ of degree $l$ and substituting the three commuting
variables by $N=2l+1$-dimensional representations of the
three-dimensional Lie algebra $\su$ allows to define a map from
functions on $\sph$ to $N\times N$ matrices, that sends Poisson
brackets to matrix commutators. The result was dubbed ``Fuzzy
Sphere'', in \cite{Mad92}. The papers \cite{KL92} prove that the
(complexified) Poisson algebra of functions on \emph{any} Riemann
surface arises as a $N\rightarrow\infty$ limit of $\glnc$ -- which had
been conjectured in \cite{BHSS91}. This result was extended to
any quantizable compact K\"{a}hler manifold in \cite{BMS94}, the
technical tool being geometric and Berezin-Toeplitz quantization.
Insight on how matrices can encode topological information (certain
sequences having been identifiable as converging to a particular
function, but $\glnc$ lacking topological invariants) was gained in
\cite{Shi04}.

The above general results, however, are merely existence proofs: there
are only two rather explicit formulas, for the two-sphere \cite{Hop82}
and for the two-torus \cite{FFZ89} (see also
\cite{Hop89}), which are quite different since the former
uses the natural embedding of the two-sphere into $\mathbb{R}^3$
whereas the latter relies on the fact that the two-torus is a quotient
$\mathbb{R}^2/\mathbb{Z}^2$.  The general results are based on the
complex nature of any compact orientable Riemann surface.

In this paper, we should like to propose an approach which to the best
of our knowledge seems to have been neglected in the literature so
far, despite its rather intuitive appeal: we are using the
`visualisable' embedding of a compact orientable Riemann surface
$\Sigma$ into $\mathbb{R}^3$ explicitly given by the set of all zeros
of a real polynomial $C$.  The function $C$ then defines a Poisson
bracket in $\mathbb{R}^3$ by the formula
\[
   \{f,g\}_C:=\nabla C \cdot\paraa{\nabla f \times \nabla g}
\]
for two real-valued smooth functions $f,g$ defined on $\mathbb{R}^3$.
Since $C$ is a Casimir function for the bracket $\{~,~\}_C$, one gets
a symplectic Poisson bracket on $\Sigma$ by restriction. The idea now
is to use the above Poisson bracket on $\mathbb{R}^3$ to first define
an infinite-dimensional non-commutative algebra as a quotient algebra
of the free algebra in three non-commuting variables by means of
relations involving $\{~,~\}_C$ and a real parameter $\hbar$. In a
second step the resulting algebra is divided by an ideal generated by
the constraint polynomial $C$ thus giving a non-commutative version of
the functions on $\Sigma$. In a third step matrix representations of
any size $N$ of this latter algebra are constructed where the
parameter $\hbar$ takes specific values depending on $N$.  It is
noteworthy that we do not depend on the zero set of $C$ to be a
regular surface. Thus, even for a singular surface (e.g., in the
transition from sphere to torus) the non-commutative analogue is still
well defined.

The main result of this paper is an explicit construction for the
the two-torus and a transition region for a two-torus emerging out of a
no longer round two-sphere following the above programme.
Encouraged by the explicit construction and by the fact that for the
two-torus our results almost coincide with the older results of
\cite{BHSS91},
we are quite optimistic that for the case of genus $g\geq 2$ this
embedding approach may give more explicit constructions than the existence
proof in \cite{KL92} and \cite{BMS94}.

\vspace{0.3cm}

\noindent The paper is organized as follows:\\
\noindent In Section \ref{sec:surface_construction} we describe
Riemann surfaces of genus $g$ embedded in $\reals^3$ as inverse images
of polynomial constraint-functions, $C(\xv)$. The above-mentioned
Poisson bracket $\{~,~\}_C$ on $\mathbb{R}^3$ is treated in Section
\ref{sec:pois_brack_r3} where the bracket restricts to a symplectic
bracket on the embedded Riemann surface $\Sigma$.

In Section \ref{sec:TorusAndSphere} step one and two of the above
programme is explicitly proven for a polynomial constraint $C$
describing the two-sphere, the two-torus, and a transition region: we
give a system of relations (eqs (\ref{eq:XY}), (\ref{eq:YZ}),
(\ref{eq:ZX})), and show that this system satisfies the hypothesis of
the Diamond lemma, thus proving that the non-commutative algebra carries
a multiplication which is a converging deformation of the pointwise
multiplication of polynomials in three
commuting variables (see Proposition \ref{prop:TorusDiamondLemma}).

In the central Section \ref{sec:RepTorusAndSphere} we completely
classify all the finite-dimensional representations of the algebras
constructed in the preceding section (two-sphere, two-torus, and
transition) which are hermitian in the sense that the variables $x$,
$y$, and $z$ are sent to hermitian $N\times N$-matrices. The main
technical tool is graph-theory describing the non-zero entries of the
matrices. Next, in Section \ref{sec:eigenvaluesequence}, we confirm
that the eigenvalue sequences of these representations reflect
topology in the sense suggested in \cite{Shi04}.

The final Section \ref{sec:GeoQuantTorus} compares the classification
results of Section \ref{sec:RepTorusAndSphere} in case of the
two-torus with the older ones obtained in \cite{BHSS91}: it is shown
that the two completely different methods give almost the same result.

\section{Genus $g$ Riemann surfaces}\label{sec:surface_construction}

\noindent The aim of this section is to present compact connected
Riemann surfaces of any genus embedded in $\reals^3$ by inverse images
of polynomials. For this purpose we use the regular value theorem and
Morse theory. Let $C$ be a polynomial in 3 variables and define
$\Sigma=C^{-1}(\{0\})$. What are the conditions on $C$, for $\Sigma$
to be a genus $g$ Riemann surface? If the restriction of $C$ to
$\Sigma$ is a submersion, then $\Sigma$ is an orientable submanifold
of $\reals^3$. $\Sigma$ has to be compact and of the desired genus.
For further details see \cite{Hir76,Hof02}.

The classification of 2 dimensional compact (connected) manifolds is
well known. In this case, there is a one to one correspondence between
topological and diffeomorphism classes. The result is that any compact
orientable surfaces is homeomorphic (hence diffeomorphic) to a sphere
or to a surface obtained by gluing tori together (connected sum). The
number $g$ of tori is called the genus and is related to the
Euler-Poincar\'{e} characteristic by the formula $\chi=2-2g$.

To compute $\chi(\Sigma)$ we apply Morse theory to a specific
function. A point $p$ of a (smooth) function $f$ on $\Sigma$ is a
singular point if $Df_p=0$, in which case $f(p)$ is a singular
value. At any singular point $p$ one can consider the second
derivative $D^2 f_p$ of $f$ and $p$ is said to be non-degenerate if
$\det(D^2 f_p)\neq 0$. Moreover one can attach an index to each such
point depending on the signature of $D^2 f$: 0 if positive, 1 if
hyperbolic and 2 if negative. A Morse function is a function such that
every singular point is non-degenerate and singular values all
distinct. Then $\chi(\Sigma)$ is given by the formula:
$$ \chi(\Sigma)=n(0)-n(1)+n(2), $$ where $n(i)$ is the number of
singular points which have an index $i$.

The $\cote$ function is defined as the restriction of the first
projection on the surface. It is not necessarily a Morse function (one
has to choose a ``good'' embedding for that), but the singular points
are those for which the gradient $\grad C$ is parallel to the $Ox$
axis. Moreover the Hessian matrix of $\cote$ at such a point $p$ is:
$$ -\frac{1}{\frac{\partial C}{\partial x}(p)}\left(
   \begin{matrix}
      \frac{\partial^2 C}{\partial y^2}(p) & \frac{\partial^2 C}{\partial y\partial z}(p)\\
      \frac{\partial^2 C}{\partial y\partial z}(p) & \frac{\partial^2 C}{\partial z^2}(p)
   \end{matrix}\right). $$
Take $$ C(\vec{x})=(P(x)+y^2)^2+z^2-\mu^2, $$ where $\mu>0$,
$P(x)=a_{2k} x^{2k}+a_{2k-1} x^{2k-1}+\cdots+a_1 x+a_0$ with
$a_{2k}>0$ and $k>0$. Obviously $\Sigma$ is closed and bounded (even
degree of $P$) hence compact. $\Sigma$ is a submanifold of $\reals^3$ if,
and only if for each $p\in \Sigma$, $DC_p\neq 0$ which is equivalent
to requiring that the polynomials $P-\mu$ and $P+\mu$ have only simple
roots. The singular points of the $\cote$ function on $\Sigma$ are the
points $(x,0,0)$ such that $P(x)^2=\mu^2$ and the Hessian matrix is:
$$ -\frac{1}{\frac{\partial C}{\partial
    x}(x,0,0)}\left(\begin{matrix}4P(x) & 0 \\ 0 &
    2\end{matrix}\right). $$ Hence it is positive or negative if, and
only if $P(x)=\mu$ and hyperbolic if, and only if $P(x)=-\mu$.  Thanks
to the fact that $P(x)$ never vanishes at a singular point, this also
shows that $\cote$ is a genuine Morse function.  Finally,
$$ n(0)+n(2)=\#\{ P = \mu \}\qquad\text{and}\qquad n(1)=\#\{ P = -\mu
\}. $$ If the polynomial $P-\mu$ has exactly 2 simple roots and the
polynomial $P+\mu$ has exactly $2g$ simple roots, then
$\chi(\Sigma)=2-2g$ and $\Sigma$ is a surface of genus $g$.
Let $g>0$. Set:
\begin{tabbing}
   \hspace{5mm}\=(i)\hspace{8mm}\=$G(t)=(t-1)(t-2^2)\ldots(t-g^2)$\hspace{1cm}\=and\hspace{1cm}\=$M=\underset{0\leq
   t\leq g^2+1}{\max} G(t),\quad \alpha\in\ \left(0,\frac{2\mu}{M}\right)$\\[2mm]
 \>(ii)\>$Q(x)=\alpha
   G(x)-\mu$\>and\>$P(x)=Q(x^2)$
\end{tabbing}
One can directly see that $Q+\mu$ has exactly $g$ simple roots, hence
$P+\mu$ has exactly $2g$ simple roots. For $t\in [0;g^2+1]$, the
function $Q(t)-\mu$ has no zero. On the other hand, for $t\geq g^2+1$,
$Q(t)-\mu$ is strictly growing and has exactly one zero. Consequently
the polynomial $P-\mu$ has exactly 2 simple roots and the surface
$\Sigma$ defined above is a genus $g$ compact Riemann surface.
Non-compact, respectively non-polynomial, higher genus Riemann surfaces
have been considered in \cite{BKL05}.

\section{The construction for general Riemann surfaces}\label{sec:pois_brack_r3}

\noindent For arbitrary smooth $C:\reals^3\longrightarrow\reals$
\begin{equation}
  \left\lbrace f,g \right\rbrace_{\reals^3} := \vec\nabla C\cdot(\vec\nabla f\times \vec\nabla g)\label{eq:pois_1}
\end{equation}
defines a Poisson bracket for functions on $\reals^3$ (see e.g. Nowak
\cite{Now97} who studied the formal deformability of
(\ref{eq:pois_1}))\footnote{While we did not (yet find a way to) use
  his results, we are very grateful for his ``New Year's Eve''
  explanations, as well as providing us with his Ph.D. Thesis.}.
Clearly, $C$ is a Casimir function of the bracket, i.e. $C$ commutes
with every function. Let now, as in Section
\ref{sec:surface_construction}, $\Sigma_g\subset\reals^3$ be described
as $C^{-1}(0)$ with
\begin{equation}
  C(\xv) = \half\paraa{P(x)+y^2}^2 + \half z^2 - c,
\end{equation}
and $c>0$. For this choice of $C$, the bracket $\pb{\cdot,\cdot}_{\reals^3}$ defines a Poisson bracket on
$\Sigma_g$ through restriction.  
The Poisson brackets between $x$,$y$ and $z$ read:
\begin{equation}\label{eq:cont_brackets}
  \begin{split}
    &\{ x,y \}_{\reals^3} = \partial_z C = z\\
    &\{ y,z \}_{\reals^3} = \partial_x C = P'(x)(P(x)+y^2)\\
    &\{ z,x \}_{\reals^3} =  \partial_y C = 2y(P(x)+y^2).
  \end{split}
\end{equation}

\noindent We claim that fuzzy analogues of $\Sigma_g$ can be obtained
via matrix analogues of \eqref{eq:cont_brackets}. Apart from possible
``explicit $1/N$ corrections'', direct ordering questions arise on the
r.h.s. of \eqref{eq:cont_brackets}, while on the l.h.s. one replaces
Poisson brackets by commutators, i.e. $\pb{\cdot,\cdot}\rightarrow
\frac{1}{i\hbar}\com{\cdot,\cdot}$.  We present the following Ansatz
for the \emph{$C$-algebra} of $\Sigma_g$, given as three relations in the
free algebra generated by the letters $X,Y,Z$:
\begin{align}
  &\com{X,Y}=i\hbar Z\label{eq:gen_XYcom}\\
  &\com{Y,Z}=i\hbar\sum_{r=1}^{2g} a_r \sum_{i=0}^{r-1} X^i\left(P(X)+Y^2\right)X^{r-1-i}\equiv\phih_X\label{eq:gen_YZcom}\\
  &\com{Z,X}=i\hbar\bracket{2Y^3+YP(X)+P(X)Y}\equiv\phih_Y\label{eq:gen_ZXcom}
\end{align}
where $\hbar$ is a positive real number and $P(X) = \sum_{r=0}^{2g}
a_r X^r$. The particular ordering in \eqref{eq:gen_YZcom} and
\eqref{eq:gen_ZXcom} is chosen such that the three equations are
consistent, in the sense of the Diamond Lemma \cite{Ber78}.
\begin{proposition}
  Let $S=\{\sigma_X,\sigma_Y,\sigma_Z\}$ be a reduction system with
  \begin{align*}
    &\sigma_X = (W_X,f_X) =\paraa{ZY,YZ-\phih_X}\\
    &\sigma_Y = (W_Y,f_Y) =\paraa{ZX,XZ+\phih_Y}\\
    &\sigma_Z = (W_Z,f_Z) =\paraa{YX,XY-i\hbar Z}.
  \end{align*}
  Then the ambiguity $(ZY)X=Z(YX)$ is resolvable if and only if
  $[X,\phih_X]+[Y,\phih_Y]=0$, and this relation is satisfied for the
  choice in \eqref{eq:gen_YZcom} and \eqref{eq:gen_ZXcom}.
\end{proposition}
\begin{proof}
  By definition, the ambiguity is resolvable if we can show that
  $A:=(YZ-\phih_X)X-Z(XY-i\hbar Z)=0$ only using the possibility to
  replace any occurrence of $W_i$ with $f_i$, for $i=X,Y,Z$. We get
  \begin{align*}
    A &= YZX-ZXY-\phih_XX+i\hbar Z^2
    = Y(XZ+\phih_Y)-(XZ+\phih_Y)Y-\phih_XX+i\hbar Z^2\\
    &=YXZ-XZY+[Y,\phih_Y] - \phih_XX + i\hbar Z^2\\
    &= (XY-i\hbar Z)Z-X(YZ-\phih_X)+[Y,\phih_Y]-\phih_XX+i\hbar Z^2
    = [X,\phih_X]+[Y,\phih_Y].
  \end{align*}
  It is then straightforward to check that
  $[Y,\phih_Y]=-[X,\phih_X]$ for the choice in \eqref{eq:gen_YZcom}
  and \eqref{eq:gen_ZXcom}.
\end{proof}

\noindent Finding explicit representations of
\eqref{eq:gen_XYcom}--\eqref{eq:gen_ZXcom}, let alone classifying
them, is of course a very complicated task. We succeeded in doing so
for $P(x)=x^2-\mu$, which corresponds to a torus when
$\mu/\sqrt{c}>1$, and deformed spheres, when $-1<\mu/\sqrt{c}<1$. In
this particular case, we were also able to construct a basis for the
quotient algebra.

\section{The torus and sphere $C$-algebras}\label{sec:TorusAndSphere}

\noindent Let us now take $P(x)=x^2-\mu$, in which case $C^{-1}(0)$, with
\begin{equation}
  C(x,y,z) = (x^2+y^2-\mu)^2+z^2-c\qquad (c>0),\label{eq:Cts}
\end{equation}
describes a torus for $\mu>\sqrt{c}$ and a sphere for
$-\sqrt{c}<\mu<\sqrt{c}$. The corresponding $C$-algebra is defined as
the quotient of the free algebra $\complex\angles{X,Y,Z}$ with the
two-sided ideal generated by the relations
\begin{align}
  &\Com{X}{Y}=i\h Z\label{eq:XY}\\
  &\Com{Y}{Z}=i\h\bracketb{2X^3+XY^2+Y^2X-2\mu X}\label{eq:YZ}\\
  &\Com{Z}{X}=i\h\bracketb{2Y^3+YX^2+X^2Y-2\mu Y}.\label{eq:ZX}
\end{align}
By introducing $W=X+iY$ and $V=X-iY$ one can rewrite \eqref{eq:YZ} and \eqref{eq:ZX} as
\begin{align}
  &\para{W^2V+VW^2}(1+\h^2)=4\mu\h^2W+2(1-\h^2)WVW\label{eq:W}\\
  &\para{V^2W+WV^2}(1+\h^2)=4\mu\h^2V+2(1-\h^2)VWV\label{eq:V}
\end{align}
and we denote by $I(\mu,\h)$ the ideal generated by these relations.
Through the ''Diamond lemma'' \cite{Ber78} one can explicitly
construct a basis of this algebra.
\begin{proposition}\label{prop:TorusDiamondLemma}
Let $\Cmuh=\complex\langle W,V\rangle\slash I(\mu,\h)$. Then a basis of $\Cmuh$ is given by
\begin{align*}
  \{V^i(WV)^j W^k\st i,j,k=0,1,2,\ldots\}.
\end{align*}
As a vector space, $\Cmuh$ is therefore isomorphic to the space of
commutative polynomials $\complex[X,Y,Z]$.
\end{proposition}
\begin{proof}
In the notation of the Diamond Lemma, let $S=\{\sigma_1,\sigma_2\}$ be a reduction system with
\begin{align*}
  \sigma_1 &= (w_{\sigma_1},f_{\sigma_1})=\para{W^2V,\frac{4\mu\h^2}{1+\h^2}W+\frac{2(1-\h^2)}{1+\h^2}WVW-VW^2}\\
  \sigma_2 &= (w_{\sigma_2},f_{\sigma_2})=\para{WV^2,\frac{4\mu\h^2}{1+\h^2}V+\frac{2(1-\h^2)}{1+\h^2}VWV-V^2W},
\end{align*}
and let $\leq$ be a partial ordering on $\langle W,V\rangle$ such that
$p<q$ if either the total degree (in $W$ and $V$) of $p$ is less than
the total degree of $q$ or if $p$ is a permutation of the letters in
$q$ and the \emph{misordering index} of $p$ is less than the
misordering index of $q$. The misordering index of a word
$a_1a_2\ldots a_k$ is defined to be the number of pairs $(a_k,a_{k'})$
with $k<k'$ such that $a_k=W$ and $a_k'=V$. This partial ordering is
\emph{compatible} with $S$ in the sense that every word in $f_{\sigma_i}$ is
less than $w_{\sigma_i}$.

We will now argue that the partial ordering fulfills the descending
chain condition, i.e. that every sequence of words such that
$w_1\geq w_2\geq\cdots$ eventually becomes constant. Assume that $w_1$
has degree $d$ and misordering index $i$. If $w_1>w_k$, then $d$ or
$i$ must decrease by at least 1. Since both the degree and the
misordering index are non-negative integers, an infinite sequence of
strictly decreasing words can not exist.

The reduction system $S$ has one \emph{overlap ambiguity}, namely,
there are two ways to reduce the word $W^2V^2$; either you write it as
$(W^2V)V$ and use $\sigma_1$, or you write it as $W(WV^2)$ and use
$\sigma_2$. In an associative algebra, these must clearly be the same,
and if they do reduce to the same expression, we call the ambiguity
\emph{resolvable}. It is now straightforward to check that the
indicated ambiguity is in fact resolvable.

The above observations allow for the use of the Diamond lemma, which
in particular states that a basis for $\Cmuh$ is given by the set of
irreducible words. In this particular case, it is clear that the words
$V^i(WV)^kW^j$ are irreducible (since they do not contain $W^2V$ or
$WV^2$) and that there are no other irreducible words.
\end{proof}
\noindent By a straightforward calculation, using \eqref{eq:W} and \eqref{eq:V}, one proves the following result.
\begin{proposition}\label{prop:Casimir}
  Define $D=WV$, $\Dt=VW$ and $\Ch=(D+\Dt-2\mu)^2+(D-\Dt)^2/\h^2$. Then it holds that
  \begin{enumerate}[(i)]
  \item $[D,\Dt]=0$,
    \item $[W,\Ch]=[V,\Ch]=0$.
  \end{enumerate}
\end{proposition}
\noindent In particular, this means that the direct non-commutative analogue of the constraint \eqref{eq:Cts} is a Casimir of $\Cmuh$.

\vspace{0.5cm}

\noindent Let us make a remark on the possibility of choosing a
different ordering when constructing a non-commutative analogue of the
Poisson algebra. Assume we choose to completely symmetrize the r.h.s
of equations \eqref{eq:cont_brackets}. Then, the defining relations of
the algebra become
\begin{align*}
  &\Com{X}{Y}=i\h Z\\
  &\Com{Y}{Z}=2i\h\bracketc{X^3+\frac{1}{3}\paraa{XY^2+Y^2X+YXY}-\mu X}\\
  &\Com{Z}{X}=2i\h\bracketc{Y^3+\frac{1}{3}\paraa{YX^2+X^2Y+XYX}-\mu Y}.
\end{align*}
Again, defining $W=X+iY$ and $V=X-iY$, gives
\begin{align*}
  &\para{W^2V+VW^2}(1+4\h^2/3)=4\mu\h^2W+2(1-2\h^2/3)WVW\\
  &\para{V^2W+WV^2}(1+4\h^2/3)=4\mu\h^2V+2(1-2\h^2/3)VWV,
\end{align*}
which, by rescaling $\h^2=\frac{3\h'^2}{3-h'^2}$, can be brought to
the form of equations \eqref{eq:W} and \eqref{eq:V}, with $\h'$ as the
new parameter.

\section{Representations of the torus and sphere algebras}\label{sec:RepTorusAndSphere}

Let us now turn to the task of finding representations $\phi$, of the
algebra $\Cmuh$, with $0<\hbar< 1$, for which $\phi(X),\phi(Y),\phi(Z)$ are hermitian
matrices, i.e. $\phi(W)^\dagger = \phi(V)$. First, we observe that any
such representation is completely reducible; hence, in the following, we need only
consider irreducible representations.

\begin{proposition}\label{prop:completely_reducible}
  Any representation $\phi$ of $\Cmuh$ such that $\phi(W)^\dagger=\phi(V)$
  is completely reducible.
\end{proposition}

\begin{proof}
  Let $\phi$ be a representation of $\Cmuh$ fulfilling the conditions
  in the proposition. Moreover, let $\A$ be the subalgebra, of the full
  matrix-algebra, generated by $\phi(W)$ and $\phi(V)$. First we note
  that since $\phi(V)=\phi(W)^\dagger$, the algebra $\A$ is invariant
  under hermitian conjugation, thus given $M\in\A$ we know that
  $M^\dagger\in\A$.

  We prove that $\Rad(\A)$ (the radical of $\A$), i.e. the
  largest nilpotent ideal of $\A$, vanishes, which implies, by the
  Wedderburn-Artin theorem, see e.g. \cite{ASS06}, that $\phi$ is completely reducible.  Let
  $M\in\Rad(\A)$. Since $\Rad(\A)$ is an ideal it follows that
  $M^\dagger M\in\Rad(\A)$. For a finite-dimensional algebra,
  $\Rad(\A)$ is nilpotent, which in particular implies that there
  exists a positive integer $m$ such that $\paraa{M^\dagger M}^m=0$.
  It follows that $M=0$, hence $\Rad(\A)=0$.
\end{proof}

\noindent In the following, we shall always assume that $\phi$ is an
\emph{hermitian} irreducible representation of $\Cmuh$. For these
representations, $\phi(D)$ and $\phi(\Dt)$ (as defined in Proposition
\ref{prop:Casimir}) will be two commuting hermitian matrices and
therefore one can always choose a basis such that they are both
diagonal. We then conclude that the value of the Casimir $\Ch$ will
always be a non-negative real number, which we will denote by $4c$.
Finding hermitian representations of $\Cmuh$ with $\phi(\Ch)=4c\mid$
thus amounts to solving the matrix equations
\begin{align}
    &(WD+\Dt W)(1+\h^2)=4\mu\h^2W+(1-\h^2)(W\Dt+DW)\label{eq:WWd}\\
    &\para{D+\Dt-2\mu\mid}^2+\frac{1}{\h^2}\para{D-\Dt}^2=4c\,\mid,\label{eq:MatrixCasimir}
\end{align}
with $D=W\Wd=\diag(d_1,d_2,\ldots,d_N)$ and $\Dt=\Wd
W=\diag(\dt_1,\dt_2,\ldots,\dt_N)$ being diagonal matrices with
non-negative eigenvalues.  The ``constraint'' \eqref{eq:MatrixCasimir}
constrains the pairs $\xv_i=(d_i,\dt_i)$ to lie on the ellipse
$(x+y-2\mu)^2+(x-y)^2/\h^2=4c$, e.g. as in Figure \ref{fig:ellipseex}.

\begin{figure}[h]
\begin{center}
  \includegraphics[height=5cm]{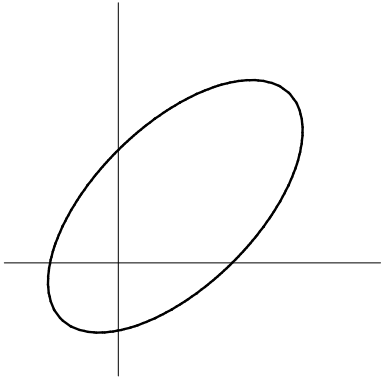}
\end{center}
\caption{The constraint ellipse.}\label{fig:ellipseex}
\end{figure}

\noindent Representations with $c=0$, which we shall call \emph{degenerate}, are
particularly simple, and can be directly characterized.
\begin{proposition}
  Let $\phi$ be an hermitian representation of $\Cmuh$
  such that $\phi(\Ch)=0$. Then $\mu\geq 0$ and there exists a unitary matrix $U$
  such that $\phi(W)=\sqrt{\mu}\,U$.
\end{proposition}
\begin{proof}
  When $D$ and $\Dt$ are non-negative diagonal matrices, $c=0$ implies
  $D=\Dt=\mu\mid$ via \eqref{eq:MatrixCasimir}, which necessarily gives $\mu\geq 0$. In this case,
  equation \eqref{eq:WWd} is identically satisfied, and we are left
  with solving the equations $W\Wd=\Wd W=\mu\mid$. Hence, there exists
  a unitary matrix $U$ such that $W=\sqrt{\mu}\,U$.
\end{proof}
\noindent Assume in the following that $c>0$. We note that any
representation $\phi'$ of $C(\mu',\hbar)$, with $\phi'(\Ch)=4c'\mid$, can be
obtained from a representation $\phi$ of $\Cmuh$ with
$\phi(\Ch)=4c\mid$, if $\mu/\sqrt{c}=\mu'/\sqrt{c'}$. Namely, one
simply defines $\phi'(W):=\sqrt[4]{c'/c}\,\,\phi(W)$.
\begin{proposition}
  Let $\phi$ be an hermitian representation of $\Cmuh$ with
  $\phi(\Ch)=4c\mid$. Then it holds that $-\sqrt{c}\leq\mu$.
\end{proposition}
\begin{proof}
  Assume that there exists a representation of $\Cmuh$ with
  $-\sqrt{c}>\mu$. Then the diagonal components of equation
  \eqref{eq:MatrixCasimir} describes an ellipse in the
  $(d,\dt)$-plane, for which all points $(d,\dt)$ satisfy that either
  $d$ or $\dt$ is strictly negative. This contradicts the fact that $D$ and
  $\Dt$ have non-negative eigenvalues. Hence, $-\sqrt{c}\leq \mu$.
\end{proof}
\noindent Writing out \eqref{eq:WWd} in components gives 
\begin{equation}\label{eq:WDcomp}
  W_{ij}\parac{\paraa{\h^2+1}(\dt_i+d_j)+\paraa{\h^2-1}(d_i+\dt_j)-4\mu\h^2}=0,
\end{equation}
and we also note that $W\Dt=DW$ yields $W_{ij}\paraa{d_i-\dt_j}=0$. If
$W_{ij}\neq 0$, the two equations give a relation between the pairs
$\xv_i=(d_i,\dt_i)$ and $\xv_j=(d_j,\dt_j)$. Namely, $\xv_j=s\para{\xv_i}$ with
\begin{align}\label{eq:sdef}
  s\paraa{d,\dt\,} = \para{4\mu\sin^2\theta+2d\cos 2\theta-\dt,d}
\end{align}
where $\h=\tan\theta$ for $0<\theta<\pi/4$.
The map $s$ is better understood if we introduce coordinates
$z(\xv)=(d-\dt)/\h$ and $\vphi(\xv)=d+\dt-2\mu$ in which case one
finds that
\begin{align}
  \begin{pmatrix}
    z\paraa{s(\xv)} \\ \vphi\paraa{s(\xv)}
  \end{pmatrix}
  =
  \begin{pmatrix}
    \cos2\theta & -\sin2\theta \\ \sin2\theta & \cos2\theta
  \end{pmatrix}
  \begin{pmatrix}
    z(\xv) \\ \vphi(\xv)
  \end{pmatrix}.
\end{align}
We conclude that $s$ amounts to a ``rotation'' on the ellipse
described by the constraint \eqref{eq:MatrixCasimir}. Let us
collect some basic facts about $s$ in the next proposition.
\begin{proposition}\label{prop:sprop}
  Let $s:\reals^2\to\reals^2$ be the map as defined above and let
  $q=e^{2i\theta}$. Then
  \begin{enumerate}[(i)]
  \item $s$ is a bijection,
  \item if $\xv(\beta_0)=\sqrt{c}\,\para{\frac{\mu}{\sqrt{c}}+\frac{\cos\beta_0}{\cos\theta},\,\frac{\mu}{\sqrt{c}}+
      \frac{\cos(\beta_0+2\theta)}{\cos\theta}}$ 
    then $s^l\paraa{\xv(\beta_0)}=\xv\para{\beta_0+2l\theta},$
  \item $s(\xv)=\xv$ if and only if $\xv=(\mu,\mu)$,
  \item if $\xv\neq (\mu,\mu)$, then $s^n(\xv)=\xv$ if and only if $q^n=1$.
  \end{enumerate}
\end{proposition}
\noindent From these considerations one realizes that it will be
important to keep track of the pairs $(i,j)$ for which $W_{ij}\neq 0$.
This leads us to a graph representation of the matrix $W$.

\subsection{Graph representation of matrices}

\noindent In this section we will introduce the directed graph of the matrix
$W$. See, e.g., \cite{FH94} for the standard terminology
concerning directed graphs.

\begin{definition}
  Let $G=(V,E)$ be a directed graph on $N$ vertices with vertex set
  $V=\{1,2,\ldots,N\}$ and edge set $E\subseteq V\times V$. We say that
  an $N\times N$ matrix $W$ is \emph{associated} to $G$ (or $G$ is
  associated to $W$) if it holds that $(ij)\in E\Leftrightarrow
  W_{ij}\neq 0$.
\end{definition}
\noindent Given an equation for $W$, we say that a graph $G$ is a
\emph{solution} if $G$ is associated to a matrix $W$, solving the
equation.  Needless to say, for a given solution $G$ there might exist
many different (matrix) solutions associated to $G$.  A graph with
several disconnected components is clearly associated to a matrix that
is a direct sum of matrices; hence, it suffices to consider connected
graphs. In the following, a solution will always refer to a
solution of \eqref{eq:WWd}.

Given a connected solution $G$, we note that given
the value of $\xv_i=(d_i,\dt_i)$, for any $i$, we can compute
$\xv_k=(d_k,\dt_k)$, for all $k$, using \eqref{eq:sdef}. Namely,
since $G$ is connected, we can always find a sequence of numbers
$i=i_1,i_2,\ldots,i_l=k$, such that $W_{i_ji_{j+1}}\neq 0$ or
$W_{i_{j+1}i_j}\neq 0$, which will give us $\xv_k=s^{m}(\xv_i)$, where
$m$ is the difference between the number of edges (in the path) directed \emph{from} $i$
and the number of edges directed \emph{towards} $i$.
\begin{proposition}
  Let $G=(V,E)$ be a connected non-degenerate solution. Then
  \begin{enumerate}[(i)]
  \item $G$ has no self-loops (i.e. $(ii)\notin E$),
    \item there is at most one edge between any pair of vertices.
  \end{enumerate}
\end{proposition}
\begin{proof}
  In both cases, assuming the opposite, it follows from
  \eqref{eq:WDcomp} that there exists an $i$ such that
  $d_i=\dt_i=\mu$. Since the graph is connected we will have
  $d_i=\dt_i=\mu$ for all $i$ ($(\mu,\mu)$ is indeed the fix-point of
  $s$), giving $c=0$. Hence, a non-degenerate solution will satisfy
  the two conditions above.
\end{proof}

\noindent Any finite directed graph has a directed cycle, which we
shall call \emph{loop}, or a directed path from a transmitter (i.e.
a vertex having no incoming edges) to a receiver (i.e. a vertex having
no outgoing edges), which we shall call \emph{string}.  The
existence of a loop or a string imposes restrictions on the
corresponding representations. From Proposition \ref{prop:sprop} we
immediately get:
\begin{proposition}\label{prop:graph_loop}
  Let $G$ be a non-degenerate solution containing a loop on $n$ vertices. Then $q^n=1$.
\end{proposition}
\begin{lemma}\label{lemma:tranrec_cond}
  Let $G$ be a solution. The vertex $i$ is a
  transmitter if and only if $\dt_i=0$. The vertex $i$ is a receiver
  if and only if $d_i=0$.
\end{lemma}
\begin{proof}
  Since $D=W\Wd$ and $\Dt=\Wd W$, we have
  \begin{align*}
    d_i &= \sum_k W_{ik}\Wb_{ik} = \sum_k |W_{ik}|^2\\
    \dt_i &= \sum_k \Wb_{ki}W_{ki} = \sum_k |W_{ki}|^2
  \end{align*}
  and it follows that $d_i=0$ if and only if $W_{ik}=0$ for all $k$,
  i.e. $i$ is a receiver. In the same way $\dt_i=0$ if and only if
  $W_{ki}=0$ for all $k$, i.e. $i$ is a transmitter.
\end{proof}
\noindent Next we prove that if $G$ is a solution, then $G$ can not
contain both a string and a loop.
\begin{lemma}\label{lemma:loop_or_string}
  Let $G$ be a non-degenerate connected solution and assume that $G$
  has a transmitter or a receiver. Then $G$ has no loop and therefore
  there exists a string.
\end{lemma}
\begin{proof}
  Let us prove the case when a transmitter exists. Let us denote the
  transmitter by $1\in V$, and by Lemma \ref{lemma:tranrec_cond} we
  have $\xv_1=(a,0)$, for some $a>0$. Assume that there exists a loop
  and let $i$ be a vertex in the loop. Since $G$ is connected
  there exists an integer $i$ such that $\xv_i = s^i(\xv_1)$. Let $l$
  be the number of vertices in the loop. From Proposition
  \ref{prop:graph_loop} we know that $q^l=1$, which means that there
  is at most $l$ different values of $\xv_k$ in the graph, and all
  values are assumed by vertices in the loop. In particular this means
  that there exists a vertex $k$ in the loop, such that $\xv_k=\xv_1$.
  But this implies, by Lemma \ref{lemma:tranrec_cond}, that $k$ is a
  transmitter, which contradicts the fact that $k$ is part of a loop.
  Hence, if a transmitter exists, there exists no loop and therefore
  there must exist a string.
\end{proof}
\noindent The above result suggests to introduce the concept of
\emph{loop representations} and \emph{string representations}, since
all representations are associated to graphs that have either a loop
or a string.

Let us now prove a theorem providing the general structure of the
representations.

\begin{theorem}\label{thm:solution}
  Let $\phi$ be an $N$-dimensional non-degenerate connected hermitian
  representation of $\Cmuh$ with $\phi(\Ch)=4c\mid$.  Then there
  exists a positive integer $k$ dividing $N$, a unitary $N\times N$
  matrix $T$, unitary $N/k\times N/k$ matrices $U_0,\ldots,U_{k-1}$
  and $\beta,\et_0,\ldots,\et_{k-1}\in\reals$ with
  $\et_1,\ldots,\et_{k-1}>0$, such that
  \begin{align}
    T\phi(W)T^\dagger &= 
    \begin{pmatrix}
      0               &  \sqrt{\et_1}\,U_1  & 0                & \cdots & 0 \\
      0               &  0                & \sqrt{\et_2}\,U_2  & \cdots & 0 \\
      \vdots          & \vdots            & \ddots           & \ddots & \vdots\\
      0               & 0                 & \cdots           & 0 &  \sqrt{\et_{k-1}}\,U_{k-1} \\
      \sqrt{\et_0}\,U_0 & 0                 & \cdots           & 0 & 0
    \end{pmatrix}\\
    \et_l &= \sqrt{c}\,\bracket{\frac{\mu}{\sqrt{c}} 
      + \frac{\cos(2l\theta+\beta)}{\cos\theta}}.
  \end{align}
\end{theorem}

\begin{proof}
  Let $U$ be a unitary $N\times N$ matrix such that $UDU^\dagger$ and $U\Dt
  U^\dagger$ are diagonal, set $\Wh=U\phi(W)U^\dagger$ and let $G$ be
  the graph associated to $\Wh$. Define $\{\xh_0,\ldots,\xh_{k-1}\}$ to
  be the set of pairwise different vectors out of the set
  $\{\xv_1,\xv_2,\ldots,\xv_N\}$, such that $\xh_{i+1}=s\para{\xh_i}$ for
  $i=0,\ldots,k-2$ (which is always possible since $G$ is connected),
  and write $\xh_i=(e_i,\et_i)$. We note that if $G$ has a
  transmitter, it must necessarily correspond to the vector $\xh_0$,
  in which case $\et_0=0$. In particular this means that no vertex
  corresponding to $\xh_i$, for $i>0$, can be a transmitter and hence,
  by Lemma \ref{lemma:tranrec_cond}, $\et_1,\ldots,\et_{k-1}>0$. Now,
  define
  \begin{align*}
    V_i = \{j\in V\st \xv_j=\xh_i\}\qquad i=0,\ldots,k-1,
  \end{align*}
  and set $l_i=|V_i|$. Since $\xh_{i+1}=s(\xh_i)$, a necessary
  condition for $(ij)\in E$ is that $j=i+1$. This implies that there
  exists a permutation $\sigma\in S_N$ (permuting vertices to give the order $V_0,\ldots,V_{k-1}$) such that
  \begin{align*}
    W' := \sigma \Wh\sigma^\dagger = 
    \begin{pmatrix}
      0               &  W_1  & 0                & \cdots & 0 \\
      0               &  0                & W_2  & \cdots & 0 \\
      \vdots          & \vdots            & \ddots           & \ddots & \vdots\\
      0               & 0                 & \cdots           & 0 &  W_{k-1} \\
      W_0 & 0                 & \cdots           & 0 & 0
    \end{pmatrix}    
  \end{align*}
  where $W_i$ is a $l_{i-1}\times l_i$ matrix (counting indices
  modulo $k$).  In this basis we get
  \begin{align*}
    D &= \diag(\underbrace{e_0,\ldots,e_0}_{l_0},\ldots,\underbrace{e_{k-1},\ldots,e_{k-1}}_{l_{k-1}})=
    W'W'^\dagger=\diag(W_1W_1^\dagger,\ldots,W_{k-1}W_{k-1}^\dagger,W_0W_0^\dagger)\\
    \Dt &= \diag(\underbrace{\et_0,\ldots,\et_0}_{l_0},\ldots,\underbrace{\et_{k-1},\ldots,\et_{k-1}}_{l_{k-1}})= 
    W'^\dagger W'=\diag(W_0^\dagger W_0,W_1^\dagger W_1,\ldots,W_{k-1}^\dagger W_{k-1}),
  \end{align*}
  which gives $W_i\Wd_i=e_{i-1}\mid_{l_{i-1}}$ and $\Wd_i W_i =
  \et_i\mid_{l_i}$.  Since $\xh_{i+1}=s(\xh_i)$ we know that
  $\et_{i+1}=e_i$, which implies that $W_i\Wd_i=\et_i\mid_{i-1}$ for
  $i=1,\ldots,k-1$. Any matrix satisfying such conditions must be a
  square matrix, i.e.  $l_i=l_{i-1}$ for $i=1,\ldots,k-1$. Hence,
  $W_i$ is a square matrix of dimension $N/k$, and there exists a
  unitary matrix $U_i$ such that $W_i=\sqrt{\et_i}U_i$. Moreover, we
  take $T$ to be the unitary $N\times N$ matrix $\sigma U$.  Finally,
  since every point $\xh_i=(e_i,\et_i)$ lies on the ellipse, there
  exists a $\beta_0$ such that $\xh_0$ corresponds to the point
  $\sqrt{c}\,(\cos(\beta_0+\theta),\sin(\beta_0+\theta))$ in the
  $(z,\vphi)$-plane, as in Proposition \ref{prop:sprop}. By defining
  $\beta=\beta_0+2\theta$, we get, since $\xh_{l+1}=s(\xh_l)$, that
  $\et_l =
  \sqrt{c}\,\bracket{\frac{\mu}{\sqrt{c}}+\frac{\cos(2l\theta+\beta)}{\cos\theta}}$.
\end{proof}

\noindent The above theorem proves the structure of any connected
representation, but the question of irreducibility still remains. We
will now prove that any representation is in fact equivalent to a
direct sum of representations where the $U_i$'s are $1\times1$-matrices.

\begin{lemma}\label{lemma:toral_equiv}
  Let $W_1$ and $W_2$ be matrices such that
  \begin{align*}
    W_1 = 
    \begin{pmatrix}
      0               &  w_1U_1  & 0                & \cdots & 0 \\
      0               &  0                & w_2U_2  & \cdots & 0 \\
      \vdots          & \vdots            & \ddots           & \ddots & \vdots\\
      0               & 0                 & \cdots           & 0 &  w_{n-1}U_{n-1} \\
      w_0U_0          & 0                 & \cdots           & 0 & 0
    \end{pmatrix}
    ;\,\,W_2 = 
    \begin{pmatrix}
      0               &  w_1\mid          & 0                & \cdots & 0 \\
      0               &  0                & w_2\mid          & \cdots & 0 \\
      \vdots          & \vdots            & \ddots           & \ddots & \vdots\\
      0               & 0                 & \cdots           & 0 &  w_{n-1}\mid \\
      w_0V            & 0                 & \cdots           & 0 & 0
    \end{pmatrix}
  \end{align*}
  where $U_0,\ldots,U_{n-1}$ are unitary matrices,
  $w_0,\ldots,w_{n-1}\in\complex$ and $V$ a diagonal
  matrix such that
  \begin{align*}
    SVS^\dagger = U_1U_2\cdots U_{n-1}U_0
  \end{align*}
  for some unitary matrix $S$. Then there exists a unitary matrix $P$
  such that
  \begin{align*}
    &W_1 = PW_2P^\dagger\quad\text{ and }\quad
    W_1^\dagger = PW_2^\dagger P^\dagger.
  \end{align*}
\end{lemma}

\begin{proof}
  Let us define $P$ as $P=\diag(S,P_1,\ldots,P_{n-1})$ with
  \begin{align*}
    P_l = (U_1U_2\ldots U_l)^\dagger S
  \end{align*}
  for $l=1,\ldots,n-1$. Then one easily checks that $W_1 =
  PW_2P^\dagger$ and $W_1^\dagger = PW_2^\dagger P^\dagger$.
\end{proof}

\noindent Note that a graph associated to a matrix such as $W_2$,
consists of $n$ components, each being either a string ($\et_0=0$) or
a loop ($\et_0>0$).  Therefore, we have the following result.

\begin{theorem}\label{thm:equiv_rep}
  Let $\phi$ be a non-degenerate hermitian representation of $\Cmuh$.
  Then $\phi$ is unitarily equivalent to a representation whose
  associated graph is such that every connected component is either a
  string or a loop.
\end{theorem}

\noindent The existence of strings or loops will depend on the ratio
$\mu/\sqrt{c}$, and therefore we split all connected representations
of $\Cmuh$ into three subsets, in correspondence with the original
surface described by the polynomial $C(x,y,z)$:
\begin{flushleft}
\begin{tabular}{lll}
  (a) & $-1<\mu/\sqrt{c}\leq 1\quad$ & -- Spherical representations\\
  (b) & $1<\mu/\sqrt{c}\leq 1/\cos\theta\quad$ & -- Critical toral representations\\
  (c) & $1/\cos\theta<\mu/\sqrt{c}\quad$ & -- Toral representations. 
\end{tabular}
\end{flushleft}

\subsection{Toral representations}

\begin{figure}[h]
  \centering
  \includegraphics[height=4cm]{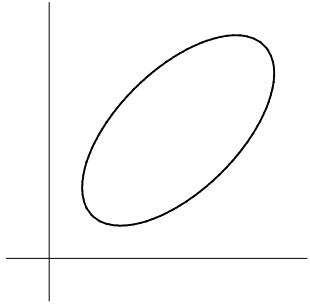}
  \caption{The constraint ellipse of a Toral representation.}
  \label{fig:toral_ellipse}
\end{figure}

\noindent For $\mu/\sqrt{c}>1/\cos\theta$ the constraint ellipse lies
entirely in the region where both $d$ and $\dt$ are strictly positive,
e.g. as in Figure \ref{fig:toral_ellipse}.  In particular this
implies, by Lemma \ref{lemma:tranrec_cond}, that a graph associated to
a toral representation can not have any transmitters or receivers.
Hence, it must have a loop, and by Proposition \ref{prop:graph_loop},
there exists an integer $k$ such that $q^k=1$. We note that the
restriction $0<\theta<\pi/4$ necessarily gives $k\geq 5$.
\begin{theorem}\label{thm:toral_rep}
  Assume that $\mu/\sqrt{c}>1/\cos\theta$ and let $k$ be a
  positive integer such that $q^k=1$. Furthermore, let
  $U_0,\ldots,U_{k-1}$ be unitary matrices of dimension $N$ and let
  $\beta\in\reals$. Then $\phi$ is an $N\cdot k$ dimensional hermitian
  toral representation of $\Cmuh$, with $\phi(\Ch)=4c\mid$, if
  \begin{equation}
    \phi(W) = 
    \begin{pmatrix}
      0               &  \sqrt{\et_1}\,U_1  & 0                & \cdots & 0 \\
      0               &  0                & \sqrt{\et_2}\,U_2  & \cdots & 0 \\
      \vdots          & \vdots            & \ddots           & \ddots & \vdots\\
      0               & 0                 & \cdots           & 0 &  \sqrt{\et_{k-1}}\,U_{k-1} \\
      \sqrt{\et_0}\,U_0 & 0                 & \cdots           & 0 & 0
    \end{pmatrix}
  \end{equation} 
  and
  \begin{equation}
    \et_l = \sqrt{c}\bracket{\frac{\mu}{\sqrt{c}}+\frac{\cos(2l\theta+\beta)}{\cos\theta}}.
  \end{equation}
\end{theorem}

\begin{definition}
  We define a \emph{single loop representation} $\phi_L$ of $\Cmuh$ to
  be a toral representation, as in Theorem \ref{thm:toral_rep}, with
  $U_i$ chosen to be $1\times 1$ matrices and $k$ to be the smallest
  positive integer such that $q^k=1$.
\end{definition}

\noindent As a simple corollary to Theorem \ref{thm:equiv_rep} we obtain

\begin{corollary}
  Let $\phi$ be a toral representation of $\Cmuh$. Then
  $\phi$ is unitarily equivalent to a direct sum of single loop representations.
\end{corollary}

\begin{proposition}
  A single loop representation of $\Cmuh$ is irreducible.
\end{proposition}

\begin{proof}
  Given a single loop representation $\phi_L$ of dimension $n$, it
  holds that $q^n=1$, and there exists no $n'<n$ such that $q^{n'}=1$,
  by definition. Now, assume that $\phi_L$ is reducible. Then, by
  Proposition \ref{prop:completely_reducible}, $\phi_L$ is equivalent
  to a direct sum of at least two representations. In particular, this
  means that there exists a toral representation of $\Cmuh$ of
  dimension $m<n$ which implies, by Proposition \ref{prop:graph_loop},
  that there exists an integer $n'<n$ such that $q^{n'}=1$. But this
  is impossible by the above argument. Hence, $\phi_L$ is irreducible.
\end{proof}

\noindent For two loop representations of the same dimension, it is not only the
value of the Casimir $\Ch$ that distinguishes them, but there is in
fact a whole set of inequivalent representations - parametrized by a
complex number.

\begin{definition}
  Let $\phi_L$ be a single loop representation in the notation of
  Theorem \ref{thm:toral_rep} with $U_l=e^{i\alpha_l}$. We define the
  \emph{index} $z(\phi_L)$ as the complex number
  \begin{align*}
    z(\phi_L) = \sqrt{\et_0\et_1\cdots\et_{k-1}}\,\,e^{i\gamma}
  \end{align*}
  with $\gamma=\alpha_0+\alpha_1+\cdots+\alpha_{k-1}$.
\end{definition}

\begin{lemma}\label{lemma:Alpermutation}
  Let $k,n$ be integers such that $\gcd(k,n)=1$ and define
  \begin{align*}
    A_l(\beta) = \cos\para{\beta+\frac{2\pi kl}{n}}
  \end{align*}
  for $l=0,1,\ldots,n-1$. Then there exists permutations
  $\sigma_+,\sigma_-\in S_n$ such that
  \begin{align*}
    A_{\sigma_+(l)}(\beta)=A_l(\beta+2\pi/n)\quad\text{ and }\quad
    A_{\sigma_-(l)}(\beta)=A_l(2\pi/n-\beta)
  \end{align*}
  for $l=0,1,\ldots,n-1$.
\end{lemma}

\begin{proof}
  Let us prove the existence of $\sigma_+$; the proof that $\sigma_-$
  exists is analogous. We want to show that there exists a permutation
  $\sigma_+$ such that $A_{\sigma_+(l)}(\beta)=A_l(\beta+2\pi/n)$. Let
  us make an Ansatz for the permutation; namely, we take it to be a
  shift with $\sigma_+(l)=l+\delta\pmod{n}$ for some
  $\delta\in\integers$. We then have to show that there exists a
  $\delta$ such that
  \begin{align*}
    \cos\para{\beta+\frac{2\pi k(l+\delta)}{n}}=
    \cos\para{\beta+\frac{2\pi(kl+1)}{n}}.
  \end{align*}
  This holds if for some $m\in\integers$
  \begin{align*}
    &\beta+\frac{2\pi k(l+\delta)}{n}=
    \beta+\frac{2\pi(kl+1)}{n}+2\pi m\equivalent\\
    &k\delta-nm=1.
  \end{align*}
  Now, can we find $\delta$ such that this holds for some $m$? It is
  an elementary fact in number theory that such an equation has integer
  solutions for $\delta$ and $m$ if $\gcd(k,n)=1$. Hence, if we set
  $\sigma_+(l)=l+\delta\pmod{n}$, where $\delta$ is such a solution,
  then the argument above shows that
  $A_{\sigma_+(l)}(\beta)=A_l(\beta+2\pi/n)$.
\end{proof}

\begin{lemma}\label{lemma:fmono}
  Let $\theta=\pi k/n$ with $\gcd(k,n)=1$, and set
  \begin{align*}
    f(\beta) = \prod_{l=0}^{n-1}
    \bracket{\mu+\frac{\sqrt{c}\cos(2l\theta+\beta)}{\cos\theta}}.
  \end{align*}
  Then $f(\beta)=f(\beta+2\pi/n)$, $f(\beta)=f(2\pi/n-\beta)$ and if
  $\beta,\beta'\in[0,\pi/n]$ then $\beta\neq\beta'$ implies that
  $f(\beta)\neq f(\beta')$.
\end{lemma}

\begin{proof}
  It follows directly from Lemma \ref{lemma:Alpermutation} that
  $f(\beta)=f(\beta+2\pi/n)=f(2\pi/n-\beta)$.

  Since $f$ is periodic, with period $2\pi/n$, it can be expanded in a Fourier series as
  \begin{align*}
    f(\beta)=\sum_{l=-\infty}^\infty a_le^{2\pi il\beta/(2\pi/n)}
    =\sum_{l=-\infty}^\infty a_le^{iln\beta}.
  \end{align*}
  Comparing the Fourier series with the original expression for $f$,
  and introducing $q=e^{2i\theta}$, we get
  \begin{align*}
    f(\beta)=\para{\frac{\sqrt{c}}{\cos\theta}}^n
    \prod_{l=0}^{n-1}\bracket{\frac{\mu\cos\theta}{\sqrt{c}}+
      \frac{1}{2}\para{q^le^{i\beta}+q^{-l}e^{-i\beta}}}
    =\sum_{l=-\infty}^\infty a_le^{inl\beta}.
  \end{align*}
  From this equality we deduce that there are only three non-zero
  coefficients in the Fourier series, namely $a_{-1},a_0,a_{1}$.
  Comparing both sides, we obtain
  \begin{align*}
    &a_{-1}=\frac{1}{2^n}q^{-n(n-1)/2}\\
    &a_{1}=\frac{1}{2^n}q^{n(n-1)/2},
  \end{align*}
  which implies that 
  \begin{align*}
    \para{\frac{\sqrt{c}}{\cos\theta}}^{-n}f(\beta)&=a_0+\frac{1}{2^n}q^{-n(n-1)/2}e^{-in\beta}
    +\frac{1}{2^n}q^{n(n-1)/2}e^{in\beta}\\
    &= a_0+\para{-\frac{1}{2}}^{n-1}\cos n\beta.
  \end{align*}
  From this it is clear that $f(\beta)\neq f(\beta')$ when
  $\beta\neq\beta'$ and $\beta,\beta'\in[0,\pi/n]$.
\end{proof}

\begin{proposition}
  Let $\phi_L$ and $\phi_L'$ be single loop representations of
  dimension $n$, such that $\phi_L(\Ch)=\phi_L'(\Ch)$. Then $\phi_L$
  and $\phi_L'$ are equivalent if and only if $z(\phi_L)=z(\phi_L')$.
\end{proposition}

\begin{proof}
  Then characteristic equation of $\phi_L(W)$ is
  $\lambda^n-z(\phi_L)$. Therefore, a necessary condition for $\phi_L$
  and $\phi_L'$ to be equivalent is that $z(\phi_L)=z(\phi_L')$. Now,
  to prove the opposite implication, assume that
  $z(\phi_L)=z(\phi_L')$. Let us denote the $\beta$ in Theorem
  \ref{thm:toral_rep} by $\beta$ and $\beta'$ for $\phi_L$ and
  $\phi_L'$ respectively. The fact that $z(\phi_L)=z(\phi_L')$ gives
  directly $\gamma=\gamma'$, and in the notation of Lemma
  \ref{lemma:fmono}, we must have $f(\beta)=f(\beta')$. By the same
  Lemma, writing $\theta=\pi k/n$, this leaves us with three
  possibilities: Either $\beta'=\beta$, $\beta'=\beta+2\pi m/n$ or
  $\beta'=2\pi m/n-\beta$ for some $m\in\integers$.  In all three
  cases, by Lemma \ref{lemma:Alpermutation}, there exists a
  permutation $\sigma$ such that for
  $W''=\sigma\phi_L'(W)\sigma^\dagger$ it holds that $\et_l''=\et_l$.
  Then it is easy to construct a diagonal unitary matrix $P$ such that
  $\phi_L(W)=P\sigma\phi_L'(W)\sigma^\dagger P^\dagger$.
\end{proof}

\noindent Hence, for a given dimension $n$ and for a given value of the
Casimir, such that toral representations exist, the set of
inequivalent irreducible representations is parametrized by a complex
number $w$ such that $\pi/n\leq|w|\leq 2\pi/n$. We relate $w$ to a
single loop representation by setting $w=\beta e^{i\gamma}$.

\subsection{Spherical representations}

\noindent In contrast to the case of toral representations, we will
show that, in a spherical representation, there can not exist any
loops. The intuitive picture is that the part of the ellipse lying in
the region where either $d$ or $\dt$ is negative, is too large to skip
by a rotation through the map $s\,$; see, e.g. Figure
\ref{fig:ellipseex}.

By Lemma \ref{lemma:tranrec_cond}, we know that the $\xv$
corresponding to a transmitter or a receiver must lie on the $d$-axis
or the $\dt$-axis respectively. For this reason, let us calculate the
points where the ellipse crosses the axes.
\begin{lemma}\label{lemma:ellipsecrossing}
  Consider the ellipse $(x+y-2\mu)^2+(x-y)/\hbar^2=4c$. Then $x=0$
  implies $y=a_\pm$ and $y=0$ implies $x=a_\pm$ with
  \begin{equation}
    a_\pm 
    =2\sin\theta\bracketb{\mu\sin\theta\pm\sqrt{c-\mu^2\cos^2\theta}}
    = 2\sin^2\theta\bracket{\mu\pm\sqrt{\mu^2+\frac{c-\mu^2}{\sin^2\theta}}\,\,}
  \end{equation}
\end{lemma}
\begin{lemma}\label{lemma:finalpoint}
  Let $\xv=(0,a_+)$, with $a_+$ as in Lemma \ref{lemma:ellipsecrossing}.
  Then $s(\xv)=(a_-,0)$.
\end{lemma}
\begin{lemma}\label{lemma:theta_n_interval}
  If $\phi$ is a spherical representation of $\Cmuh$, that contains a
  string on $n$ vertices, then
  \begin{align}
    0<(n+1)\theta\leq\pi.
  \end{align}
\end{lemma}
\begin{proof}
  Let us denote the vectors corresponding to the vertices in the
  string by $\xv_1,\ldots,\xv_n$ and we define $0<\beta,\theta_0<2\pi$
  through $\xv_1=\xv(\beta)$ and $\xv_n=\xv(\beta+\theta_0)$ in the
  notation of Proposition \ref{prop:sprop}. Since
  $\xv_n=s^{n-1}\paraa{\xv(\beta)}$ we must have that
  $(n-1)2\theta=\theta_0+2\pi k$ for some integer $k\geq 0$. Let us
  prove that $k=0$.  For a spherical representation, $a_-\leq 0$,
  which implies, by Lemma \ref{lemma:finalpoint}, that
  $s\paraa{\xv(\beta+\theta_0)}=(a_-,0)$ can not correspond to a
  vertex of a connected representation. Hence, for any
  $\alpha\in(0,2\theta)$, $s\paraa{\xv(\beta+\theta_0-\alpha)}$ can
  not correspond to a vertex of a connected representation. This
  implies that $k=0$, i.e. the string never crosses the $\dt$-axis.
  Therefore $0<(n-1)2\theta=\theta_0<2\pi$.  Again, by Lemma
  \ref{lemma:finalpoint}, both vectors $s(0,a_+)$ and $s^2(0,a_+)$
  have non-positive components which implies that
  $0<(n+1)2\theta\leq2\pi$. In fact, equality is attained when $a_-=0$.
\end{proof}
\begin{proposition}
  Let $\phi$ be a spherical representation of $\Cmuh$. Then the
  associated graph has no loops.
\end{proposition}
\begin{proof}
  In the same way as in the proof of Lemma
  \ref{lemma:theta_n_interval}, we can argue that for
  $\alpha\in(0,2\theta)$, $s\paraa{\xv(\beta+\theta_0-\alpha)}$ has a
  negative component (or equals $(0,0)$), which implies that it is
  impossible to have loops.
\end{proof}

\noindent Hence, we have excluded the possibility of loop
representations and can conclude that all spherical representations
are string representations. We therefore get the following corollary to
Theorem \ref{thm:equiv_rep}.

\begin{corollary}
  Let $\phi$ be a spherical representation of $\Cmuh$. Then
  $\phi$ is unitarily equivalent to a direct sum of string representations.  
\end{corollary}

\noindent Let us now investigate the conditions for
the existence of strings.

\begin{lemma}\label{lemma:stringcond}
  Let $\xv_1=(a,0)$ and $\xv_n=(0,b)$. Then $s^{n-1}(\xv_1)=\xv_n$ if and only if
  \begin{enumerate}[(i)]
  \item $q^n=-1$, $\mu=0$ and $a=b$
  \item $q^n=1$ and $b=-a+4\mu\sin^2\theta$ 
  \item $q^n\neq\pm 1$ and 
    \begin{equation}
      a=b=-\frac{2\mu\sin\theta\sin(n-1)\theta}{\cos n\theta}
    \end{equation}
  \end{enumerate}
  In particular, if $a=a_+$ and $q^n=1$, then $b=a_-$.
\end{lemma}

\begin{proposition}
  Let $\phi$ be a spherical representation of $\Cmuh$ containing a string on $n$ vertices. Then
  \begin{align}\label{eq:cntheta}
    \sqrt{c}\cos n\theta + \mu\cos\theta = 0.
  \end{align}
\end{proposition}
\begin{proof}
  Assume the existence of a string on $n$ vertices. From Lemma
  \ref{lemma:stringcond} we can exclude the possibility that $q^n=1$,
  since $a_-<0$. Hence, either $q^n=-1$ and $\mu=0$ or $q^n\neq\pm 1$.
  If $q^n=-1$ and $\mu=0$ then \eqref{eq:cntheta} is clearly
  satisfied. Now, assume $q^n\neq\pm 1$ and
  $a=b=\frac{2\mu\sin\theta\sin(n-1)\theta}{\cos n\theta}$.  Demanding
  that $(a,0)$ and $(0,b)$ lie on the ellipse determines $c$ as
  $c=\mu^2\cos^2\theta/\cos^2 n\theta$. Let us set
  $\varepsilon=\sgn{\mu}$. Recalling that $0<(n+1)\theta\leq\pi$, from
  Lemma \ref{lemma:theta_n_interval}, demanding $a>0$ makes it
  necessary that $\sgn(\cos n\theta)=-\varepsilon$, which determines
  the sign of the root in the statement.
\end{proof}

\noindent As we have seen, the existence of a loop puts a restriction
on $\hbar$ through the relation $q^n=1$. For the case of strings, the
restriction comes out as a restriction on the possible values
of the Casimir.

In the next theorem we show that the necessary conditions for the
existence of spherical representations are in fact sufficient.

\begin{theorem}\label{thm:sphere_rep}
  Let $n$ be a positive integer, $c$ a positive real number such that
  $\sqrt{c}\cos n\theta+\mu\cos\theta=0$ and $0<(n+1)\theta\leq\pi$.
  Furthermore, let $U_1,\ldots,U_{n-1}$ be $N\times N$ unitary
  matrices. Then $\phi$ is a $N\cdot n$-dimensional spherical
  representation of $\Cmuh$, with $\phi(\Ch)=4c\mid$, if
  \begin{align*}
    \phi(W) = 
    \begin{pmatrix}
      0               &  \sqrt{\et_1}\,U_1  & 0                & \cdots & 0 \\
      0               &  0                & \sqrt{\et_2}\,U_2  & \cdots & 0 \\
      \vdots          & \vdots            & \ddots           & \ddots & \vdots\\
      0               & 0                 & \cdots           & 0 &  \sqrt{\et_{n-1}}\,U_{n-1} \\
      0               & 0                 & \cdots           & 0 & 0
    \end{pmatrix}
  \end{align*}
  and
  \begin{align*}
    \et_l = \frac{2\sqrt{c}\sin{l\theta\sin(n-l)\theta}}{\cos\theta}.
  \end{align*}
\end{theorem}

\begin{proof}
  It is easy to check that the matrix $\phi(W)$ satisfy
  \eqref{eq:WWd}, since $s(\et_l,\et_{l-1})=(\et_{l+1},\et_l)$.
  Moreover, it is clear that $\et_l>0$ since $0<(n-1)\theta<\pi$. Let
  us show that it is indeed a spherical representation, i.e.
  $-1<\mu/\sqrt{c}\leq 1$. Since $\sqrt{c}\cos
  n\theta+\mu\cos\theta=0$, we get that
  \begin{align*}
    \frac{\mu}{\sqrt{c}} = -\frac{\cos n\theta}{\cos\theta}
  \end{align*}
  and from $0<(n+1)\theta\leq\pi$ we obtain $0<n\theta\leq\pi-\theta$. From
  this it follows that $|\cos n\theta|\leq|\cos\theta|$ which implies
  that $\phi$ is a spherical representation.
\end{proof}

\begin{definition}
  We define a \emph{single string representation} $\phi_S$ of $\Cmuh$ to be a
  spherical representation, as in Theorem \ref{thm:sphere_rep}, with
  $U_i$ chosen to be $1\times 1$ matrices. 
\end{definition}

\begin{proposition}
  Any single string representation of $\Cmuh$ is irreducible.
\end{proposition}

\begin{proof}
  Assume that $\phi_S$ is reducible and has dimension $n$ with
  $\phi_S(\Ch)=4c\mid$. Then, by Proposition
  \ref{prop:completely_reducible}, $\phi_S$ is equivalent to a direct
  sum of at least two representations of dimension $<n$. In
  particular, this implies that there exists a representation $\phi$
  of dimension $m<n$ with $\phi(\Ch)=4c\mid$. But this is false, since
  there is at most one integer $l$ such that
  $\xv(\beta+2l\theta)=\xv(\beta+\theta_0)$, for $0<(l+1)2\theta<2\pi$ and
  $0<\theta_0<2\pi$.
\end{proof}

\noindent We conclude that the single string representations are the
only irreducible spherical representations. Moreover, two single string
representations $\phi_S$ and $\phi_S'$, of the same dimension, are equivalent if and only if
$\phi_S(\Ch)=\phi_S'(\Ch)$.

\subsection{Critical toral representations}

\begin{figure}[h]
  \centering
  \includegraphics[height=5cm]{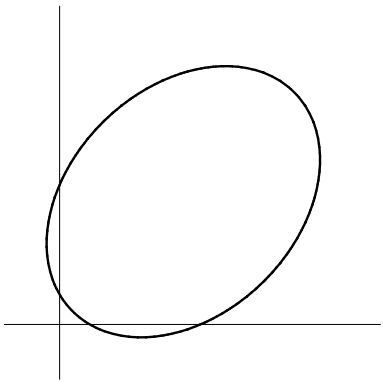}
  \caption{The constraint ellipse of a critical toral representation.}
  \label{fig:ellipse_crittoral}
\end{figure}

\noindent In the case of critical toral representations, the
constraint ellipse intersects the positive $d$ (resp. $\dt$) axis
twice, as in Figure \ref{fig:ellipse_crittoral}. As we will show,
there are both loop representations and string representations. String
representations can be obtained from Theorem \ref{thm:sphere_rep}, by
demanding that $1<\mu/\sqrt{c}\leq 1/\cos\theta$ instead of
$0<(n+1)\theta\leq\pi$. Let us as well give an example of a loop
representation.

\begin{proposition}
  Assume that $\theta=\pi/N$, $N\geq 5$ odd and $1<\mu/\sqrt{c}\leq
  1/\cos\theta$. If we define $\phi$ as in Theorem \ref{thm:toral_rep}
  with $\beta=0$, then $\phi$ is a critical toral representation of
  $\Cmuh$.
\end{proposition}

\begin{proof}
  One simply has to check that 
  \begin{align*}
    \et_l=\sqrt{c}\bracket{\frac{\mu}{\sqrt{c}}+\frac{\cos(2l\frac{\pi}{N})}{\cos\frac{\pi}{N}}} > 0
  \end{align*}
  for $l=0,\ldots,N-1$. If $N$ is odd then $2\theta
  l\notin(\pi-\theta,\pi+\theta)$ and $2\theta
  l\notin(2\pi-\theta,2\pi)$, which implies that $|\cos2\theta
  l|<|\cos\theta|$. Since $\mu/\sqrt{c}>1$ we conclude that $\et_l>0$
  for $l=0,\ldots,N-1$.
\end{proof}

\noindent In contrast to the previous cases, it is, for a given value
of the Casimir, possible to have both string representations and loop
representations. Namely, if we assume that $q^n=1$ and let $\xv_1$
correspond to the largest intersection with the $d$-axis, then
$s^{n-1}(\xv_1)$ will be the smallest intersection with the $\dt$-axis
(cp. Lemma \ref{lemma:stringcond}), and one can check that all pairs
$\xv_i$, for $i=2,\ldots,n-1$ will be strictly positive.

\section{Eigenvalue distribution and surface topology}\label{sec:eigenvaluesequence}

In \cite{Shi04}, by using arguments similar to those in the WKB
approximation in quantum mechanics, a geometric approach is introduced
in the matrix regularization.  General matrix elements of a matrix are
related to entities computed from the corresponding function on the
surface.  In particular, it has been shown that Morse theoretic
information of topology manifests itself in certain branching
phenomena of eigenvalue distribution of a single matrix.

\begin{figure}
\centering
\psfrag{xlabel}[][]{x label}
\psfrag{i}[tc][tc]{$i$}
\psfrag{m1}[tc][tc]{$\mu=0.9$}
\psfrag{m2}[tc][tc]{$\mu=1.1$}
\psfrag{m3}[tc][tc]{$\mu=1.3$}
\psfrag{e}[bc][cr][1][90]{$\lambda_i$}
\psfrag{d}[bc][cr][1][90]{$\lambda_{i+1} -\lambda_i $}
\includegraphics[width=13cm]{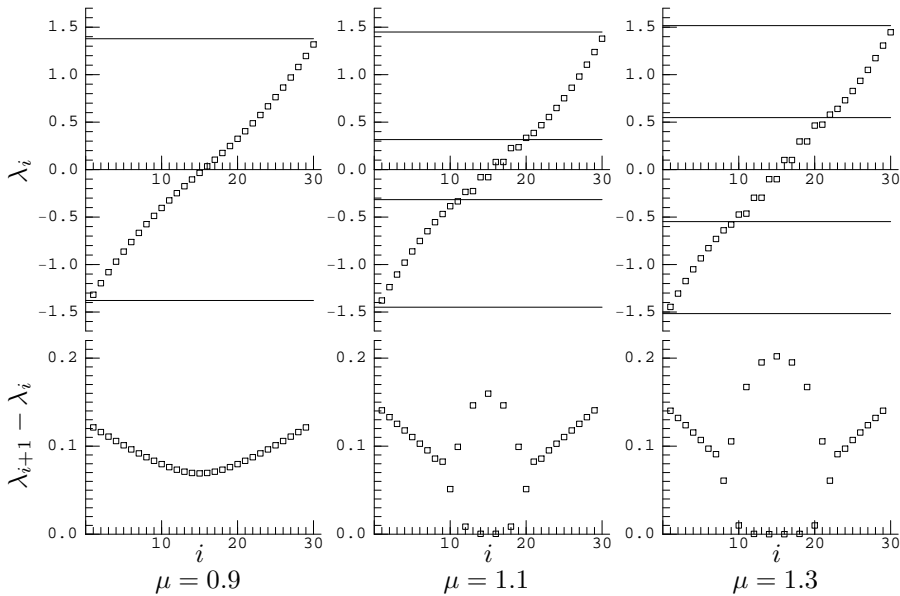}
\caption{
Plot of $\lambda_i$ and$\lambda_{i+1} -\lambda_i$ versus $i$,
where ${\lambda_1 < \lambda_2< \ldots <\lambda_N}$ are eigenvalues of $X$,
for $\mu=0.9, 1.1, 1.3$. The size of matrices is given by $N=30$.
Critical values of $x$ are also shown by the horizontal lines. }
\label{FigEigenvaluePlots}
\end{figure}

Eigenvalues of $X$ (whose continuum counterpart, $x$, is a Morse
function on the surface) in the representations obtained in Section
\ref{sec:RepTorusAndSphere}, do exhibit this branching phenomena, as
is consistent with the results in \cite{Shi04}.  In Figure
\ref{FigEigenvaluePlots}, eigenvalues of $X$, computed numerically,
for the case $\mu=0.9,1.1,1.3$ are shown. (We use the normalization
convention in which $c=1$, so that the transition between sphere and
torus occurs at $\mu=1$. The size of matrices is given by $N=30$.
For the toral representation, we have taken the
additional ``phase shift'' parameter $\beta$ to be zero. Using
different $\beta$'s does not change the plot qualitatively.)  The
horizontal lines correspond to the critical values of the function $x$
on the surface.

The plots directly reflect the Morse theoretic information of
topology, with $x$ as the Morse function, for each case
$\mu=0.9,1.1,1.3$.  For the case $\mu=0.9$, there are two critical
values which are connected by a single branch. Correspondingly,
the eigenvalue plot shows that there is only one ``sequence'' of
eigenvalues $\lambda_1 < \lambda_2 <\ldots <\lambda_N$ which increase
smoothly.  For the cases $\mu=1.1$ and $\mu=1.3$, there are four
critical values of $x$, say $x_A < x_B < x_C <x_D$. For $x_A<x<x_B$
and $x_C<x<x_D$ the surface consists of single branch, whereas for
$x_B<x<x_C$, the surface consists of two branches. Correspondingly, in
the plot of eigenvalues, one sees that eigenvalues $x_A< \lambda_i
<x_B$ and $x_C <\lambda_i < x_D$ each consists of single smoothly
increasing eigenvalue sequence, whereas eigenvalues $x_C < \lambda_i <
x_D$ is naturally divided into two sequences both of which increase
smoothly.  This branching phenomena of eigenvalues can be seen more
manifestly if one plots the difference between eigenvalues,
$\lambda_{i+1}-\lambda_{i}$, as is shown in the figure.  For more
detailed discussion about the eigenvalue sequences and its branching,
see \cite{Shi04}.  From the figure it can also be seen that by
decreasing the parameter $\mu$ from $1.3$ to $1.1$ the part of the
surface which have two branches shrinks, as is consistent with the
geometrical picture about the transition between torus and sphere.

\section{Comparison with the Berezin-Toeplitz quantization}\label{sec:GeoQuantTorus}

The purpose of this section is to compare matrix representations
obtained in Section \ref{sec:RepTorusAndSphere}, in the torus case,
with those one gets using Berezin-Toeplitz quantization. Full details
and proofs can be found in \cite{Hof07}. We shall use Theorem 5.1 from
the paper \cite{BHSS91} applied to $\Sb^1\times \Sb^1$. Namely $n=1$,
$\tau=1$ and we omit the Laplacian terms:
\begin{equation*}
\frac{\pi}{m}\sum_{k=1}^n \tau_k\left(r_k^2+\frac{r^2_{k+n}}{t^2_k}\right)
\quad\mathrm{and}\quad \prod_{s=1}^n\exp\left(-\frac{\pi\tau_s}{2m}\left(r^2_s+\frac{r^2_{s+n}}{\tau^2_s}\right)\right)
\end{equation*}
We reformulate it for simplicity and to fix notations.
\begin{thm}\label{thm:bhss}
  Let $r_1,r_2\in\Zb$ and $N\geq 5$ an integer. Then the $N\times
  N$-matrix corresponding to the face function $e^{2\pi
    i(r_1\theta+r_2\varphi)}$ is:
\begin{equation*}
 \mat\left(e^{2\pi i(r_1\theta+r_2\varphi)}\right) = \chi^{r_1 r_2} S^{-r_1} T^{r_2}\and \chi := e^{-\frac{\pi i}{N}}
\end{equation*}
where the $S$ and $T$ are matrices such that:
\begin{equation*}
 S = \smat{1}{1}{1}{1},\quad T=\diag(1,q,\ldots,q^{N-1})\ou q := \chi^2 = e^{-\frac{2\pi i}{N}}.
\end{equation*}
\end{thm}

\begin{rmk}
  The $\mat$ map is not a morphism of algebras. However, $\mat$ is
  continuous for the topology of uniform convergence.
\end{rmk}

\noindent To apply this theorem to the torus case, i.e. the regular
values of the polynomial function $(x^2+y^2-\mu)^2+z^2-\nu^2$ (with
$\mu/\nu>1$), one has to choose the right embedding:

\begin{prop}
 Let $\mu,\nu\in\Rb$ such that $\mu/\nu>1$. By using the parametrization:
\begin{equation*}
\left\{
\begin{array}{l}
  x(\theta,\varphi) = \cos(2\pi\theta)\sqrt{\nu\cos(2\pi\varphi)+\mu}\\
  y(\theta,\varphi) = \sin(2\pi\theta)\sqrt{\nu\cos(2\pi\varphi)+\mu}\\
  z(\theta,\varphi) = \nu\sin(2\pi\varphi)
\end{array}
\right.
\end{equation*}
one gets:
\begin{eqnarray}
 \mat(x) & = & \frac{S}{2} \sqrt{\fatone\mu + \frac{\nu}{2}\left(\chi^{-1} T+ \chi T^{-1}\right)} + \frac{S^{-1}}{2} \sqrt{\fatone\mu + \frac{\nu}{2}\left(\chi T+\chi^{-1}T^{-1}\right)}\label{mat_x}\\
 \mat(y) & = & \frac{S}{2i} \sqrt{\fatone\mu + \frac{\nu}{2}\left(\chi^{-1} T+ \chi T^{-1}\right)} - \frac{S^{-1}}{2i} \sqrt{\fatone\mu + \frac{\nu}{2}\left(\chi T+\chi^{-1}T^{-1}\right)}\label{mat_y}\\
 \mat(z) & = & \frac{\nu}{2i}\left(T-T^{-1}\right)\label{mat_z}
\end{eqnarray}
\end{prop}
\begin{proof}
  The key idea is an expansion in Fourier series of
  $\sqrt{\mu+\nu\cos(2\pi\varphi)}$. We then replace face functions by
  matrices $T$ and $S$ according to Theorem \ref{thm:bhss}. Square
  roots of matrices are well defined since the matrices are positive
  definite.
\end{proof}

\begin{lem}\label{lem:diag}
 Let $D=\diag(d_1,\dots,d_N)$ be a diagonal $N\times N$-matrix, then:
\begin{equation*}
S^{-1} D S=\diag(d_N,d_1,\dots,d_{N-1}) \and S D S^{-1} = \diag(d_2,\dots,d_N,d_1).
\end{equation*}
\end{lem}

\noindent Let us denote:
\begin{equation*}
 D :=  \sqrt{\fatone\mu + \frac{\nu}{2}\left(\chi T+\chi^{-1}T^{-1}\right)}\and \til D := \sqrt{\fatone\mu + \frac{\nu}{2}\left(\chi^{-1} T+ \chi T^{-1}\right)}.
\end{equation*}
Then one can write (\ref{mat_x}) and (\ref{mat_y}) as:
\begin{equation*}
 \mat(x) = \frac{1}{2} \left( S \til D + S^{-1} D\right) \and \mat(y) = -\frac{i}{2}\left(S \til D - S^{-1} D\right).
\end{equation*}
It is easily seen that the matrices $D$ and $\til D$ are diagonal:
\begin{eqnarray*}
 D = \diag\left(\sqrt{\mu+\nu\cos\left(\frac{2\pi l}{N}+\frac{\pi}{N}\right)}\right)_{l=1,\dots,N}\\ \til D = \diag\left(\sqrt{\mu+\nu\cos\left(\frac{2\pi l}{N}-\frac{\pi}{N}\right)}\right)_{l=1,\dots,N}.
\end{eqnarray*}
By Lemma \ref{lem:diag},
\begin{equation*}
S \til D = S \til D S^{-1} S = \diag\left(\sqrt{\mu+\nu\cos\left(\frac{2\pi l}{N}+\frac{\pi}{N}\right)}\right)_{l=1,\dots,N} \times S = D S.
\end{equation*}
As a consequence, $\mat(x)$ and $\mat(y)$ can be written as:
\begin{equation*}
 \mat(x) = \frac{1}{2} \left( D S + S^{-1} D\right) \and \mat(y) = -\frac{i}{2}\left(D S - S^{-1} D\right).
\end{equation*}

\begin{thm}\label{thm:berezin_toeplitz}
The matrices $\mat(x)$, $\mat(y)$ and $\mat(z)$ are:
\begin{align*}
 & \mat(x) =  
  \frac{1}{2}\spsmat{x_1}{x_2}{x_{N-1}}{x_N},\\
  & \mat(y) = -\frac{i}{2}\smsmat{y_1}{y_2}{y_{N-1}}{y_N},\\
  & \mat(z) = \diag(z_1,z_2,\ldots,z_N)
\end{align*}
where the $x_l$'s, $y_l$'s and $z_l$'s (for $l=1,\dots,N$) are:
\begin{equation*}
 x_l = y_l = \sqrt{\mu+\nu\cos\left(\frac{2\pi l}{N}+\frac{\pi}{N}\right)}\and z_l = -\nu\sin\left(\frac{2\pi l}{N}\right).
\end{equation*}
\end{thm}

\noindent These matrices satisfy the following relations:

\begin{thm}\label{thm:berezin_toeplitz2}
 Let $\mu,\nu\in\Rb$ and $N\geq 5$ such that $\mu/\nu>1$. If one assumes $\h=\tan(\theta)$ with $\theta:=\pi/N$, then:
 \begin{align*}
  & [X,Y] = i\h (\cos(\theta) Z)\\
  & [Y,(\cos(\theta) Z)] = i\h \left(X(X^2+Y^2-\mu\fatone)+(X^2+Y^2-\mu\fatone)X\right)\\
  & [(\cos(\theta) Z),X] = i\h \left(Y(X^2+Y^2-\mu\fatone)+(X^2+Y^2-\mu\fatone)Y\right)\\
  & (X^2+Y^2-\mu\fatone)^2 + (\cos(\theta) Z)^2 = (\nu\cos(\theta))^2\fatone.
 \end{align*}
 where $X:=\mat(x)$, $Y:=\mat(y)$ and $Z:=\mat(z)$ are the matrices
 obtained in the theorem \ref{thm:berezin_toeplitz}. Let us stress that
 $\theta$ is not related to an angle of a parametrization.
\end{thm}
\begin{proof}
 This is a direct computation on matrices.
\end{proof}

\noindent Hence, one can see that the matrices $\mat(x),\mat(y)$ and
$\mat(z)$ from Theorem \ref{thm:berezin_toeplitz} look very similar to
those from Theorem \ref{thm:toral_rep}, if one chooses $\beta=\pi/N$. Moreover the
free parameter $\nu$ can be set to 1 or to $1/\cos(\pi/N)$.
Consequently these matrices are asymptotically equal. The relations
that the matrices from Theorem \ref{thm:berezin_toeplitz} satisfy are also
very similar to \eqref{eq:XY}-\eqref{eq:ZX}.\vspace{5mm}

\noindent\textbf{Acknowledgement.}
We would like to thank the Swedish Research Council, the Royal
Institute of Technology, the Japan Society for the Promotion of
Science, the Albert Einstein Institute, the Sonderforschungsbereich
``Raum-Zeit-Materie'', the ESF Scientific Programme MISGAM, and the
Marie Curie Research Training Network ENIGMA for financial support
resp. hospitality.


\newpage


\end{document}